%% file: kdd2022_yahoo.tex
\documentclass[sigconf]{acmart}

\usepackage{bm}
\usepackage{amsmath}

\AtBeginDocument{%
  \providecommand\BibTeX{{%
    \normalfont B\kern-0.5em{\scshape i\kern-0.25em b}\kern-0.8em\TeX}}}

\begin{document}

\title{An Analysis of Selection Bias Issue for Online Advertising}

\author{Shinya Suzumura}
\email{ssuzumur@yahoo-corp.jp}
\affiliation{%
  \institution{Yahoo Japan Corporation}
  \city{Tokyo}
  \country{JP}
}

\author{Hitoshi Abe}
\email{hitoabe@yahoo-corp.jp}
\affiliation{%
  \institution{Yahoo Japan Corporation}
  \city{Tokyo}
  \country{JP}
}

\begin{abstract}
  In online advertising, a set of potential advertisements can be ranked by a certain auction system
  where usually the top-1 advertisement would be selected and displayed at an advertising space.
  In this paper, we show a selection bias issue that is present in an auction system.
  We analyze that the selection bias destroy \emph{truthfulness} of the auction,
  which implies that the buyers (advertisers) on the auction can not maximize their profits.
  Although selection bias is well known in the field of statistics and there are lot of studies for it,
  our main contribution is to combine the theoretical analysis of the bias with the auction mechanism.
  In our experiment using online A/B testing,
  we evaluate the selection bias on an auction system whose ranking score is the function of predicted CTR (click through rate) of advertisement.
  The experiment showed that the selection bias is drastically reduced by using a multi-task learning which learns the data for all advertisements.
\end{abstract}


\maketitle

\input{sec1_introduction}
\input{sec2_related_works}
\input{sec3_main}
\newpage
\input{sec4_experiments}
\input{sec5_conclusion}

\bibliographystyle{ACM-Reference-Format}
\bibliography{bibliography}

\end{document}

%% file: sec1_introduction.tex
\section{Introduction}
\label{sec:introduction}
Online advertising \cite{he2014practical,zhu2017optimized,jin2018real,grigas2017profit,vasile2017cost}
is an effective way to maximize advertiser utilities which are related to ROI (return on investment), ROAS (return on advertising spend), and so on.
The objective of advertising depends on the advertisers:
one may want to acquire new customers,
while others may want to boost brand recognition of their products.
Online advertising can be regarded as a resource allocation problem
\cite{johari2004efficiency,katoh1998resource,shi2015faster}
or budget allocation problem
\cite{grigas2017profit,staib2017robust,hatano2016adaptive,maehara2015budget,soma2014optimal,zhang2014optimal,karande2013optimizing,han2020contextual,abrams2006revenue}.
The resources are activities of users who visit a media which is depicted in \figurename~\ref{fig:yahoo_news}, for instance.
Each advertiser wants to effectively allocate their budget for each of the user activities.
\begin{figure}[t]
  \centering
  \includegraphics[width=\linewidth]{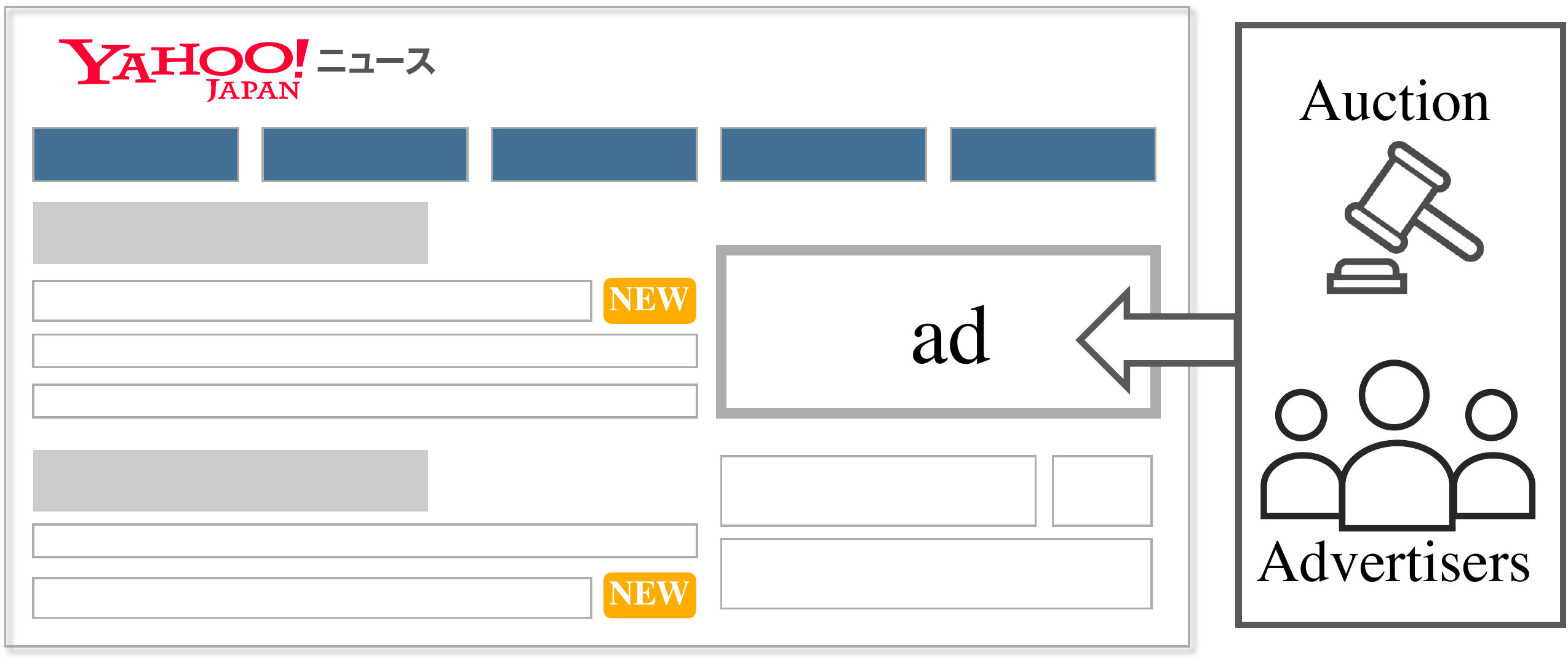}
  \caption{
    An example of media which has an advertising space.
    “ad” is an advertisement selected by a certain auction.
  }
  \label{fig:yahoo_news}
\end{figure}
Note that advertising is competitive since the resources are finite.
In general, the winner of competitors is determined by a certain auction.
Here we explain GSP (generalized second-price auction) \cite{edelman2007internet}
that is still one of the most standard auctions in online advertising \cite{karande2013optimizing}.

\subsection*{GSP (Generalized Second-Price Auction)}
Here we provide a short description of GSP.
In this paper, we focus on CPC (cost per click) advertising in which the advertisers set the (sealed) bid price\footnote{In CPC advertising, the bid price is a price that advertiser is willing to pay for user click.
Moreover, auction in online advertising is generally a sealed-bid auction, which means that each of the advertisers can not know the bid prices of the other advertisers.} for user click.
Let us consider $m$ advertisements which are ranked by the following ranking scores
\begin{align}
  \label{eq:ranking_scores}
  \text{Bid}_{(1)} \text{CTR}_{(1)}
  \ge
  \text{Bid}_{(2)} \text{CTR}_{(2)}
  \ge \cdots \ge
  \text{Bid}_{(m)} \text{CTR}_{(m)},
\end{align}
where the subscripted number with parentheses is the rank of ad (advertisement),
$\text{Bid}_{(i)}$ and $\text{CTR}_{(i)}$ are the bid price and click through rate\footnote{CTR is a probability of ad click for a given user.} for the rank $i$-th ad, respectively.
In GSP,
the original bid price is adjusted by the ranking score of the second rank ad,
i.e.,
\begin{align}
  \label{eq:cpc_in_2nd_price_auction}
  \text{CPC}_{(1)} =
    \frac{
      \text{Bid}_{(2)} \text{CTR}_{(2)}
    }{
      \text{CTR}_{(1)}
    }
    \le \text{Bid}_{(1)},
\end{align}
where $\text{CPC}_{(1)}$ is cost per click (or equally, spend for a click) of the advertiser who has the highest rank ad.
CPC computation (\ref{eq:cpc_in_2nd_price_auction}) indicates that
the expected cost for the highest rank ad is
\begin{align*}
  \text{CPC}_{(1)} \text{CTR}_{(1)} = \text{Bid}_{(2)} \text{CTR}_{(2)},
\end{align*}
which only depends on the second rank ad,
and thus
$\text{CPC}_{(1)}$ is equal to the original bid price $\text{Bid}_{(1)}$ if and only if $\text{Bid}_{(1)} \text{CTR}_{(1)} = \text{Bid}_{(2)} \text{CTR}_{(2)}$.
This implies that there is no need to decrease the original bid price manually.

In this paper, we consier the single item auction where only the top-1 ad is selected and displayed at a position within the media\footnote{If we have multiple positions for advertising in the media, then we use the single item auction in each of the positions.}.
Under the single item setting, GSP is inherently the same as Vickrey–Clarke–Groves (VCG) auction
which is a \emph{truthful} auction \cite{aggarwal2006truthful,levin2004auction} under the assumption that advertisers have a large amount of budgets,
this means that the best strategy is to set bid price as the maximum price that advertisers are willing to pay for user click.
Although the optimal bidding strategy to maximize advertiser utilities is simple in truthful auctions,
the difficulty is that CTR is unknown in practice and
CPC computation (\ref{eq:cpc_in_2nd_price_auction}) is sensitive to the error of estimated CTR (see \cite{feng2019online} for a regret analysis for the estimation error).
Furthermore, we suffer from a selection bias issue \cite{heckman1979sample}
in CPC computation as follows.

\subsection*{Selection Bias Issue}
Since CTR is unknown, we need to estimate or predict it in practice.
Let $\widehat{\text{CTR}}$ be an estimated CTR, then (\ref{eq:ranking_scores}) can be rewritten as
\begin{align}
  \label{eq:ranking_scores_with_estimated_CTR}
  \text{Bid}_{(1)^\prime} \widehat{\text{CTR}}_{(1)^\prime}
  \ge
  \text{Bid}_{(2)^\prime} \widehat{\text{CTR}}_{(2)^\prime}
  \ge \cdots \ge
  \text{Bid}_{(m)^\prime} \widehat{\text{CTR}}_{(m)^\prime},
\end{align}
where the subscript is the rank using estimated CTRs.
The above ordered statistics are affected by a selection bias which is written by
\begin{align}
  \label{eq:selection_bias}
  \mathbb{E}\left[
    \widehat{\text{CTR}}_{(i)^\prime}
  \right]
  = b_{(i)^\prime} \text{CTR}_{(i)^\prime},
  ~~\text{for~any~fixed}~ i \in \{1,\cdots,m\},
\end{align}
where $b_{(i)^\prime} \in \mathbb{R}$ is the selection bias.
In this paper, we reveal that the selection biases hold the property $b_{(1)^\prime} \ge b_{(2)^\prime}$ in a certain assumption,
and thus (\ref{eq:cpc_in_2nd_price_auction}) with the estimated CTRs would be over discounted as follows:
\begin{align}
  \label{eq:cpc_with_selection_bias}
  \frac{
    \text{Bid}_{(2)^\prime} \mathbb{E}\left[ \widehat{\text{CTR}}_{(2)^\prime} \right]
  }{
    \mathbb{E}\left[ \widehat{\text{CTR}}_{(1)^\prime} \right]
  } =
  \frac{
    \text{Bid}_{(2)^\prime} b_{(2)^\prime} \text{CTR}_{(2)^\prime}
  }{
    b_{(1)^\prime} \text{CTR}_{(1)^\prime}
  }.
\end{align}
Since each selection bias $b_{(i)^\prime}$ depends on the bid prices of all advertisements on the auction (see \S\ref{sec:theoretical_analysis} for more details),
we conjuncture that the auction would be no longer truthful\footnote{Since CPC is over discounted by controlling the bid price,
we do not know the optimal bidding strategy under the selection bias issue.}.
To the best of our knowledge, there are no truthful auctions under the selection bias issue.
Thus, we wish to compute the unbiased estimator such that
$\mathbb{E}\left[\widehat{\text{CTR}}_{(i)^\prime}\right] = \text{CTR}_{(i)^\prime}$.
The details of this problem is discussed in \S\ref{sec:theoretical_analysis}.

\subsection*{Summary of Our Main Contribution}
Our main contribution is to provide a theoretical analysis of the selection bias in (\ref{eq:selection_bias}).
It is worth mentioning that the biases in (\ref{eq:cpc_with_selection_bias}) are canceled
if $b_{(1)^\prime}$ and $b_{(2)^\prime}$ are the same.
We reveal that the selection biases $b_{(1)^\prime}$ and $b_{(2)^\prime}$ are NOT the same,
which means that GSP using the estimated CTRs is strictly affected the selection bias issue.

In our experiment using online A/B testing,
we evaluate the selection bias on the auction and we show that it is drastically reduced by using a multi-task learning that learns the data for all advertisements.

%% file: sec2_related_works.tex
\section{Related Works}
\label{sec:related_works}
Although we focus on GSP under the single item setting in this paper,
another common auction is first-price auction whose CPC equals to the original bid price.
However, it is typically unstable and high variant --
in the paper of \cite{edelman2007strategic},
it was shown that the transition of CPC in a first-price based auction has a distinctive “sawtooth” pattern.
Furthermore,
we can not ignore the selection bias issue even if we consider a first-price based auction,
because we need to predict the minimum winning price of the auction to maximize advertiser utilities;
this is inherently the same as the prediction of the second-price.
The selection bias issue for the winning price prediction can be regarded as \emph{the Winner’s Curse} \cite{lee2018winner,xu2011bayesian,zollner2007overcoming}.

In recent studies \cite{lee2014exact,lee2016exact,tibshirani2016exact} for statistical inference or hypothesis testing,
one way to correct the selection bias is to consider the maximum likelihood estimator with the conditional distribution of ordered statistics.
In the following, we consider the conditional distribution of the ordered statistic
$\widehat{\text{CTR}}_{(1)^\prime}$
in (\ref{eq:ranking_scores_with_estimated_CTR}).
Let $\widehat{\text{CTR}}_{i}$ be the estimated CTR for the $i$-th ad.
For any $j_k \in \{1,\cdots,m\} \setminus \{i\}$,
the event that $\widehat{\text{CTR}}_{i}$ equals to $\widehat{\text{CTR}}_{(1)^\prime}$ can be written as
\begin{align*}
  \text{Bid}_{i} \widehat{\text{CTR}}_{i}
  &\ge
  \text{Bid}_{j_1} \widehat{\text{CTR}}_{j_1}, \\
  \text{Bid}_{i} \widehat{\text{CTR}}_{i}
  &\ge
  \text{Bid}_{j_2} \widehat{\text{CTR}}_{j_2}, \\
  &\vdots \\
  \text{Bid}_{i} \widehat{\text{CTR}}_{i}
  &\ge
  \text{Bid}_{j_{m-1}} \widehat{\text{CTR}}_{j_{m-1}},
\end{align*}
or equally
\begin{align*}
  \left[
    \!\!
    \begin{tabular}{ccccc}
      $\text{Bid}_{i}$ & $-\text{Bid}_{j_1}$ & $0$ & $\cdots$ & $0$ \\
      $\text{Bid}_{i}$ & $0$ & $-\text{Bid}_{j_2}$ & $\cdots$ & $0$ \\
      && $\vdots$ && \\
      $\text{Bid}_{i}$ & $0$ & 0 & $\cdots$ & $-\text{Bid}_{j_{m-1}}$
    \end{tabular}
    \!\!
  \right]
  \left[
    \begin{tabular}{l}
      $\widehat{\text{CTR}}_{i}$ \\
      $\widehat{\text{CTR}}_{j_1}$ \\
      $\vdots$ \\
      $\widehat{\text{CTR}}_{j_{m-1}}$
    \end{tabular}
    \!\!
  \right]
  \ge
  \left[
    \!\!
    \begin{tabular}{c}
      $0$ \\
      $0$ \\
      $\vdots$ \\
      $0$
    \end{tabular}
    \!\!
  \right].
\end{align*}
The above linear constraints are summarized in the form of $A\bm{y} \ge \bm{0}$
where $\bm{y}$ is $m$-dimensional vector for the (non-ordered) estimated CTRs and $A$ is a fixed matrix which is independent of $\bm{y}$,
and we define $y_1 = \widehat{\text{CTR}}_{i}$.
By using this,
the event that $\widehat{\text{CTR}}_{i} = \widehat{\text{CTR}}_{(1)^\prime}$ can be rewritten as $A\bm{y} \ge \bm{0}$.
According to the recent study of \cite{lee2016exact},
if the vector $\bm{y}$ is generated from a multivariate normal distribution,
then the conditional distribution of ordered statistic $\widehat{\text{CTR}}_{(1)^\prime}$ is computable
because it is the distribution of $y_1$ on the convex polytope $A\bm{y} \ge \bm{0}$.
Therefore, the unbiased estimator for the unknown true $\text{CTR}_{i}$ in the case that
$\widehat{\text{CTR}}_{i} = \widehat{\text{CTR}}_{(1)^\prime}$
is the conditional maximum likelihood estimator with the distribution
$\text{P}(y_1 | A\bm{y} \ge \bm{0})$.

In general, we have no prior knowledge of $\bm{y}$, which means that the conditional distribution can not be analytically computed.
In order to overcome this problem, a sampling based approach was proposed in \cite{nie2018adaptively};
the conditional maximum likelihood estimator is approximately computed by using Gumbel-max trick \cite{gumbel1954statistical}
where the true $\text{CTR}_{i}$ is estimated by adaptively collecting data on a certain multi-armed bandit algorithm.
The selection bias has been also studied in causal inference \cite{pearl2009causal},
and the basic idea of bias correction is to use IPW (inverse probability weighting) \cite{robins1994estimation}
or so called \emph{doubly robust estimator} \cite{bang2005doubly,saito2020doubly}.
Our main contribution is inherently different from the above studies.
We analyze the relation of selection biases of multiple ordered statistics,
and we show that the effect of the selection biases is not canceled in GSP
even if the estimated CTRs are originally unbiased.

%% file: sec3_main.tex
\section{Main Contribution}
\label{sec:theoretical_analysis}
In this section,
we first show a synthetic experiment about the selection bias explained in (\ref{eq:selection_bias}).
Second, we analyze the selection bias.
The experiment links to intuitive understanding of our analysis.

\subsection*{Simulation Study}
\label{subsec:simulation_study}
We consider the situation where there are two advertisements whose true CTRs are $\text{CTR}_1$ and $\text{CTR}_2$, respectively.
Let $\widehat{\text{CTR}}_i, ~ i \in \{1,2\}$ be the estimated CTRs as
\begin{align}
  \label{eq:CTR_generation_process}
  \widehat{\text{CTR}}_i = \frac{c_i}{n_i},~~
  c_i \sim \text{Bin}(\text{CTR}_i, n_i),
\end{align}
where $\text{Bin}(\cdot)$ is Binomial distribution,
$c_i$ is the number of observed ad clicks,
and
$n_i$ is the number of ad impressions\footnote{The number of ad impressions is the number of times when ad is displayed within the media.}.
In this section, we assume that $n_i$ is a fixed variable,
then $\widehat{\text{CTR}}_i$ is the unbiased estimator such that $\mathbb{E}\left[ \widehat{\text{CTR}}_i \right] = \text{CTR}_i$.
For simplicity, we assume that the bid prices of both advertisements are set to ones which implies that CPC defined in (\ref{eq:cpc_in_2nd_price_auction}) must be $\text{CTR}_{(2)} / \text{CTR}_{(1)}$.
In the following simulation, we check whether the observed
$\text{CPC}_{(1)^\prime} = \widehat{\text{CTR}}_{(2)^\prime} / \widehat{\text{CTR}}_{(1)^\prime}$ is equal to the expected
$\text{CPC}_{(1)} = \text{CTR}_{(2)} / \text{CTR}_{(1)}$.

We compute the observed $\text{CPC}_{(1)^\prime}$
when the parameters are summarized in Table~\ref{tb:parameter_settings_simulation}.
Table~\ref{tb:observed_CPC_simulation} shows the average of observed CPC over independent 20000 simulations in each parameter setting.
This result indicates that the observed $\text{CPC}_{(1)^\prime}$ tends to be over discounted
when the number of ad impressions is small or $\text{CTR}_1,\text{CTR}_2$ are competitive.
This means that we suffer from the selection bias issue even if the estimated CTRs are originally unbiased.
\begin{table}[b]
  \caption{Several parameter settings used in (\ref{eq:CTR_generation_process}).}
  \label{tb:parameter_settings_simulation}
  \centering
  \begin{tabular}{c|c|l}
  \hline
  setting & $n_1=n_2$ & \multicolumn{1}{c}{true CTRs} \\ \hline
  (a) & $\phantom{0}5000$ & $\text{CTR}_1=0.05,\text{CTR}_2=0.05$ \\
  (b) & $\phantom{0}5000$ & $\text{CTR}_1=0.05,\text{CTR}_2=0.045$ \\
  (c) & $\phantom{0}5000$ & $\text{CTR}_1=0.05,\text{CTR}_2=0.04$ \\ \hline
  (d) & $20000$ & $\text{CTR}_1=0.05,\text{CTR}_2=0.05$ \\
  (e) & $20000$ & $\text{CTR}_1=0.05,\text{CTR}_2=0.045$ \\
  (f) & $20000$ & $\text{CTR}_1=0.05,\text{CTR}_2=0.04$ \\ \hline
  \end{tabular}
\end{table}
\begin{table}[b]
  \caption{
  The average of observed $\text{CPC}_{(1)^\prime}$ over 20000 trials, which is denoted as $\overline{\text{CPC}}_{(1)^\prime}$.
  The expected $\text{CPC}_{(1)}$ using true CTRs.
  }
  \label{tb:observed_CPC_simulation}
  \centering
  \begin{tabular}{c|c|c|c}
  \hline
  setting & $\text{CPC}_{(1)}$ & $\overline{\text{CPC}}_{(1)^\prime}$ & $\overline{\text{CPC}}_{(1)^\prime} / \text{CPC}_{(1)}$ \\ \hline
  (a) & 1.0 & 0.934 & 0.934 \\
  (b) & 0.9 & 0.894 & 0.993 \\
  (c) & 0.8 & 0.803 & 1.00 \\ \hline
  (d) & 1.0 & 0.966 & 0.966 \\
  (e) & 0.9 & 0.900 & 1.00 \\
  (f) & 0.8 & 0.800 & 1.00 \\ \hline
  \end{tabular}
\end{table}
\figurename~\ref{fig:sampling_distribution_CPC} shows the sampling distributions of observed
$\text{CPC}_{(1)^\prime}$ in each parameter setting:
these distributions are seriously skewed in the settings (a),(d), while they are almost symmetric in the settings (c),(f).
\figurename~\ref{fig:sampling_distribution_ordered_statistics} shows the sampling distributions of ordered statistics
$\widehat{\text{CTR}}_{(1)^\prime}$ and $\widehat{\text{CTR}}_{(2)^\prime}$:
these distributions are overlapped in the setting (a), while they are almost non-overlapped in the setting (f).

These phenomena are interpretable as follows: it is because an overestimated CTR tends to be
$\widehat{\text{CTR}}_{(1)^\prime}$
rather than
$\widehat{\text{CTR}}_{(2)^\prime}$.
In other words, if the ranking is deterministic as
the probability
$\text{P}(\widehat{\text{CTR}}_1 \ge \widehat{\text{CTR}}_2) = 1$,
then
$\widehat{\text{CTR}}_{(1)^\prime}$ reduces to $\widehat{\text{CTR}}_{1}$
regardless of whether $\widehat{\text{CTR}}_{1}$ is overestimated or not --
this means that $\widehat{\text{CTR}}_{(1)^\prime}$ would be the unbiased estimator such that
$\mathbb{E}\left[ \widehat{\text{CTR}}_{(1)^\prime} \right] = \mathbb{E}\left[ \widehat{\text{CTR}}_1 \right] = \text{CTR}_1$.
\begin{figure}[!t]
  \centering
  \begin{tabular}{ccc}
    \includegraphics[width=0.475\linewidth]{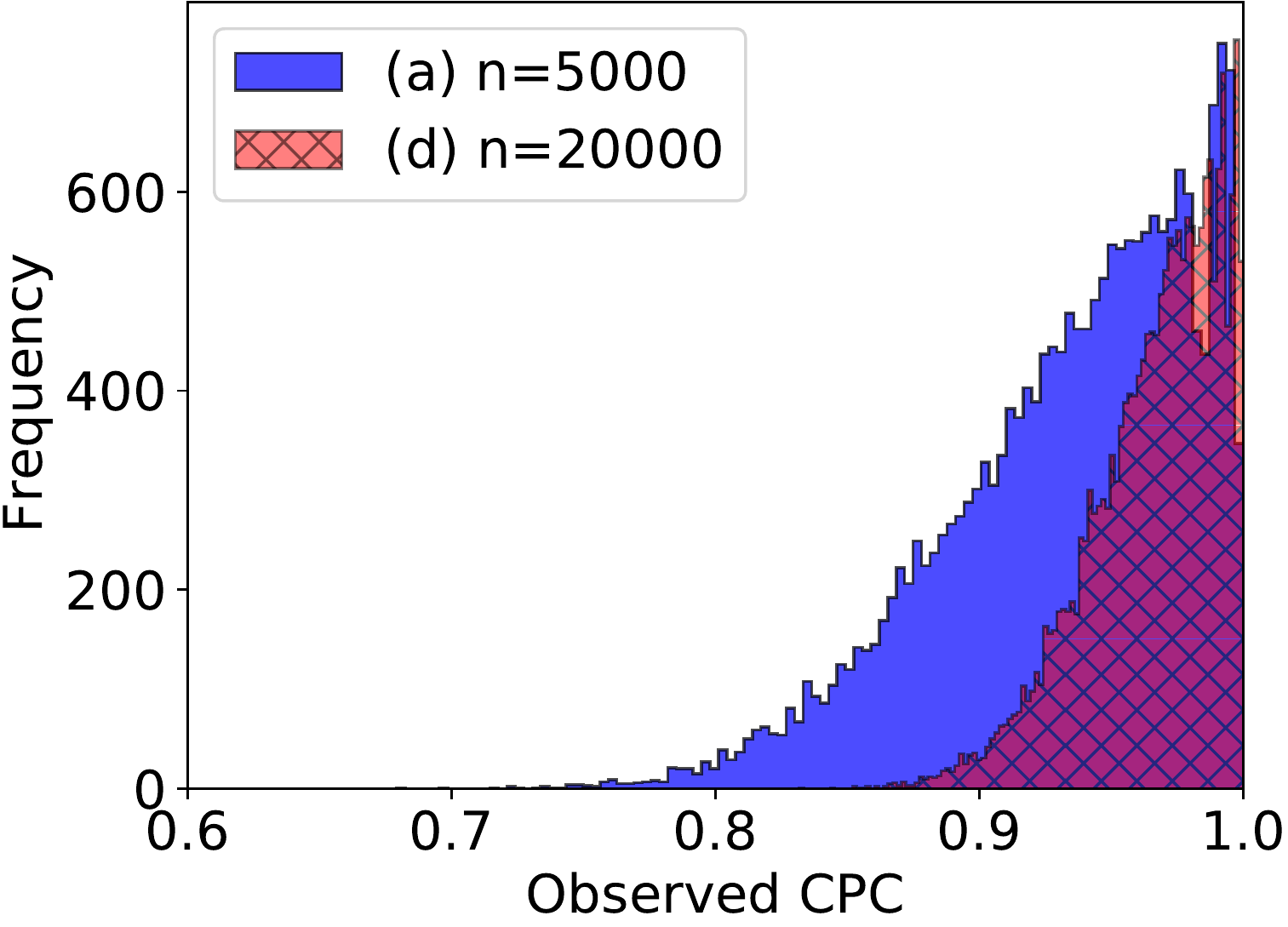} &
    \includegraphics[width=0.475\linewidth]{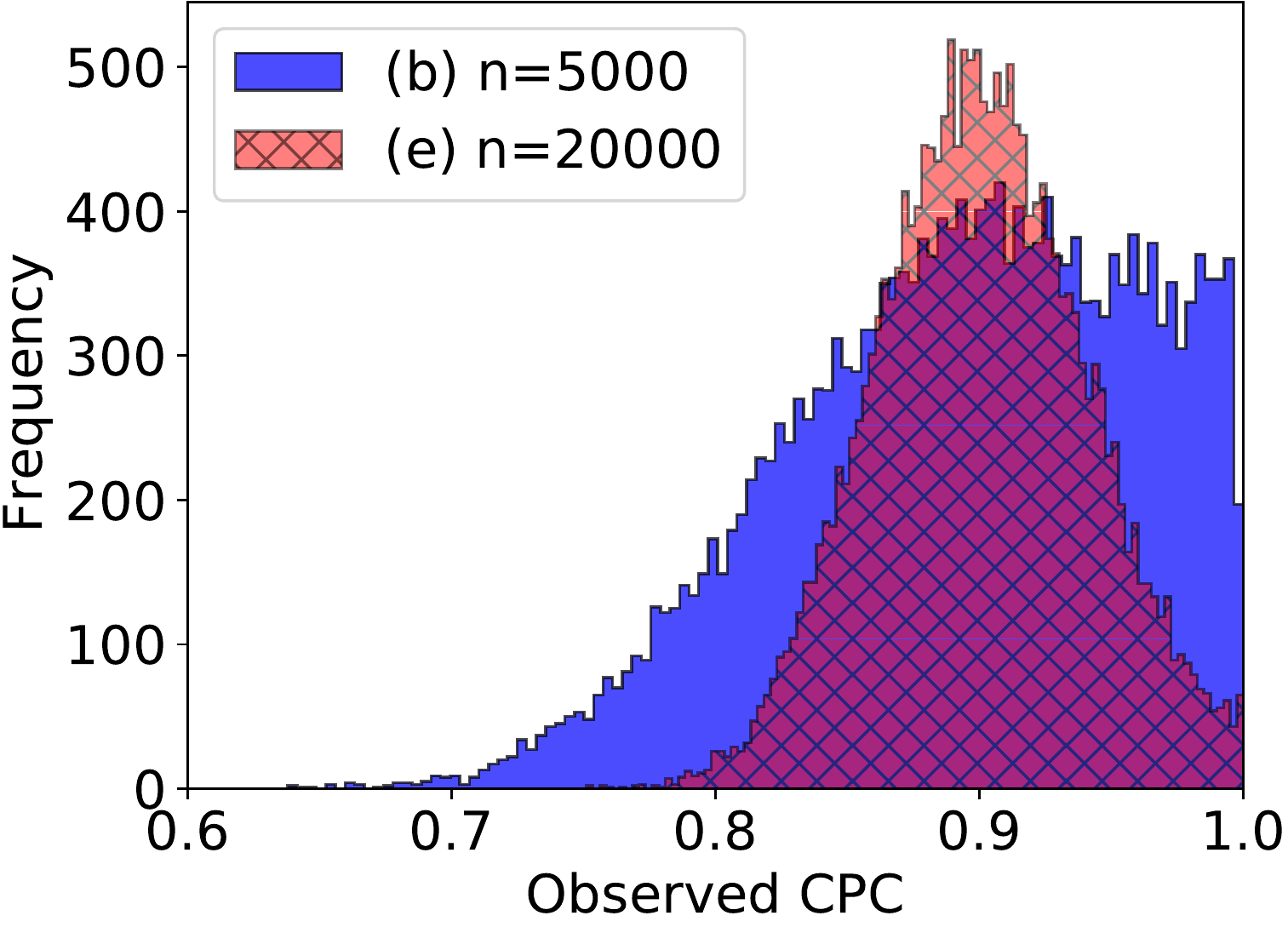} \\
    parameter settings: (a),(d) & parameter settings: (b),(e)
  \end{tabular}
  \includegraphics[width=0.475\linewidth]{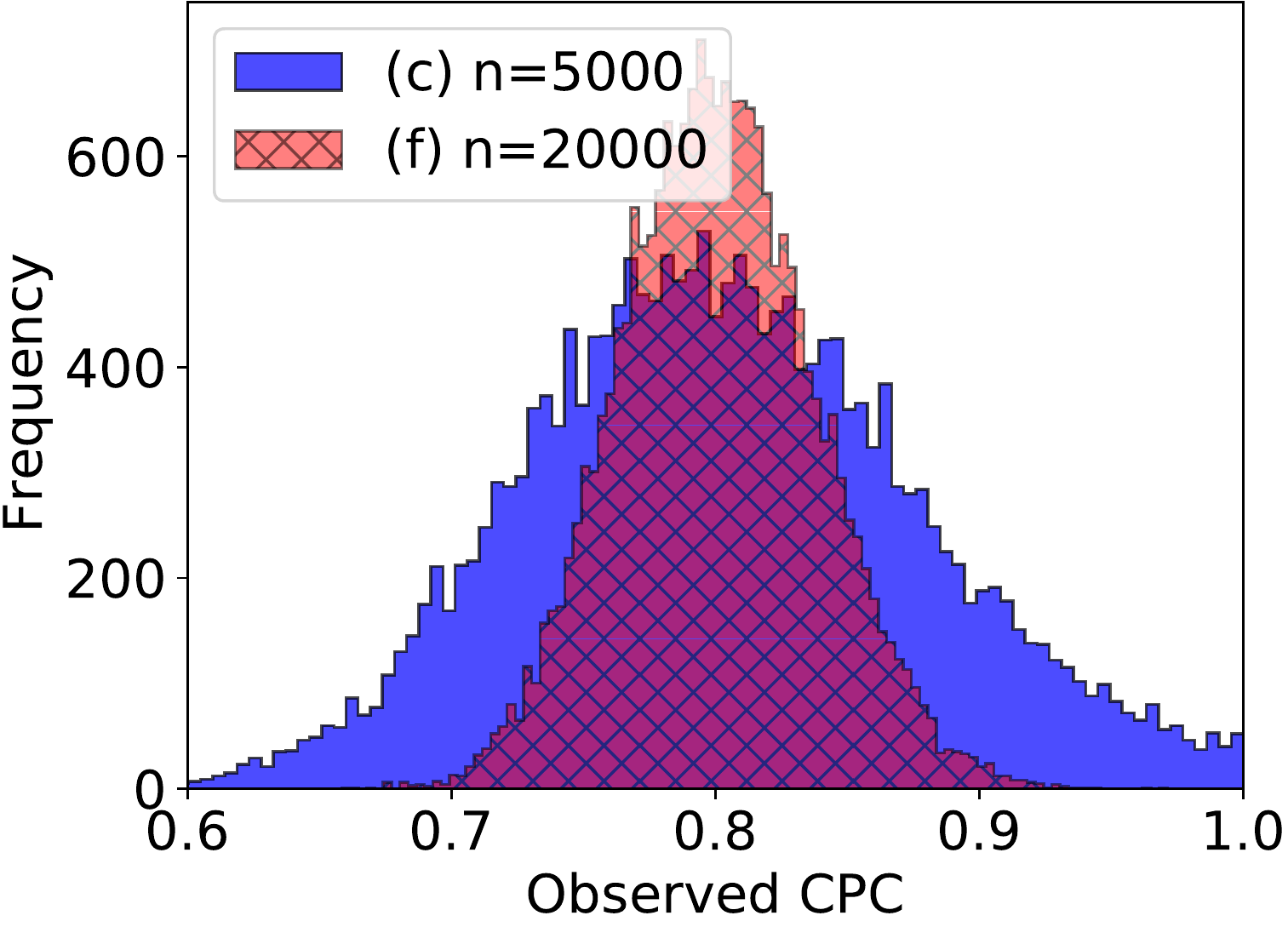} \\
  parameter settings: (c),(f)
  \caption{
    The sampling distributions of observed CPCs,
    which are simulated by independent 20000 trials.
  }
  \label{fig:sampling_distribution_CPC}
\end{figure}
\begin{figure}[!t]
  \centering
  \begin{tabular}{ccc}
    \includegraphics[width=0.475\linewidth]{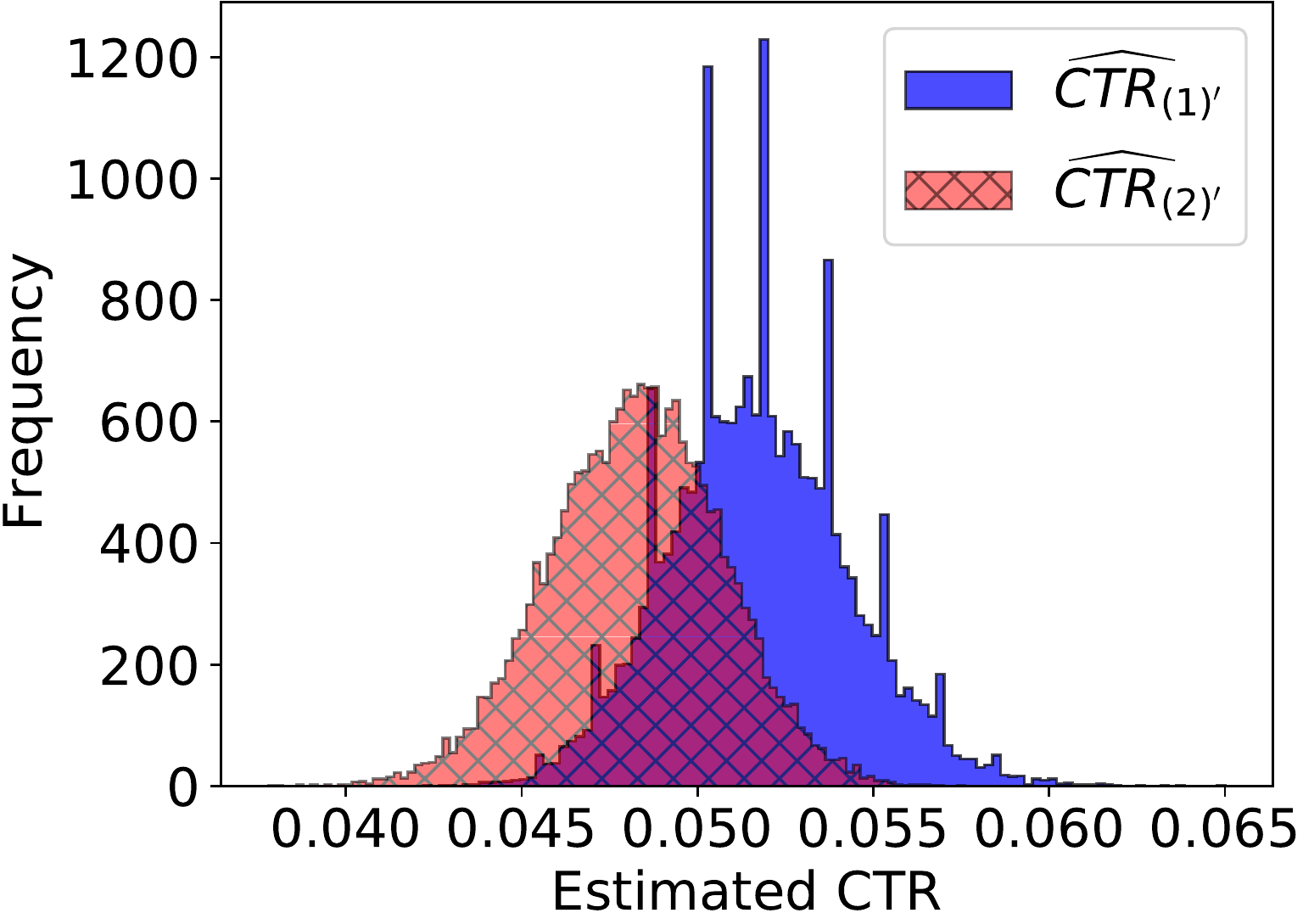} &
    \includegraphics[width=0.475\linewidth]{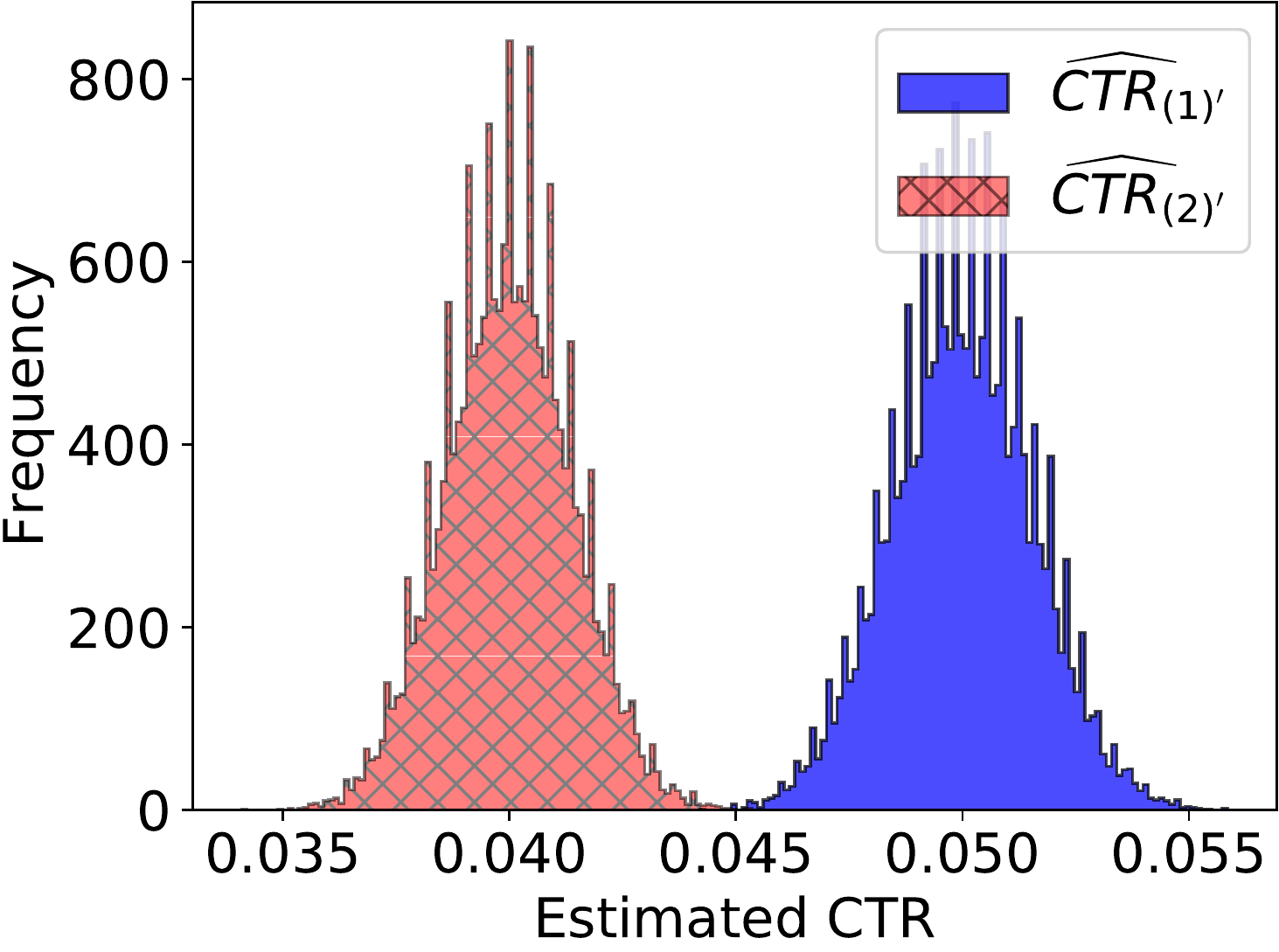} \\
    parameter setting (a) & parameter setting (f)
  \end{tabular}
  \caption{
    The sampling distributions of ordered statistics $\widehat{\text{CTR}}_{(1)^\prime}$ and $\widehat{\text{CTR}}_{(2)^\prime}$,
    which are simulated by independent 20000 trials.
  }
  \label{fig:sampling_distribution_ordered_statistics}
\end{figure}

\subsection*{Theoretical Analysis}
Here we show that the selection bias depends on the rank of ad
which implies that the biases $b_{(1)^\prime}$ and $b_{(2)^\prime}$ are not canceled in (\ref{eq:cpc_with_selection_bias}).
In the remainder, we denote $r_i = r(\text{ad}_i)$ and $s_i = \text{Bid}_{i} \widehat{\text{CTR}}_{i}$
as the rank and ranking score of the $i$-th ad, respectively.
Since $\widehat{\text{CTR}}_{i}$ is a random variable,
we denote $R_i,S_i$ as random variables for $r_i,s_i$ where the small numbers are observed values.
Moreover, let us denote $\text{P}_{S_i}(s)$ as PDF (probability density function) with respect to the ranking score for the $i$-th ad.
\begin{definition}
  We call the following functions $f(x), g(x)$ as \emph{splittable} at $v \in \mathbb{R}$,
  \begin{align}
    \label{eq:splittable}
    \exists v, \forall x \geq 0, f(v - x) \leq g(v - x) \land f(v + x) \geq g(v + x).
  \end{align}
\end{definition}
An example of the splittable functions can be seen in
\figurename~\ref{fig:sampling_distribution_ordered_statistics},
where the conditional sampling distributions $\text{P}_{S_i}(s_i | r_i = 1)$ and $\text{P}_{S_i}(s_i | r_i = 2)$
are roughly splittable at $s_i = 0.05$ or $0.045$.
\begin{theorem}
  \label{theorem:main_contribution}
  For any $i \in \{1,\cdots,m-1\}$, $k \in \{1,\cdots,m-1\}$, and non-negative ranking score $s_i \ge 0$,
  if the conditional PDFs $f(s_i) = {\rm P}_{S_i}(s_i | r_i = k)$ and $g(s_i) = {\rm P}_{S_i}(s_i | r_i = k+1)$ are splittable functions defined in (\ref{eq:splittable}),
  then the expected value of $S_i$ conditional on the rank $r_i$ is relatively large when the rank is higher, i.e.,
  \begin{align}
    \label{eq:main_contribution}
    \mathbb{E} \left[ S_i | r_i = k \right] \ge \mathbb{E} \left[ S_i | r_i = k+1 \right].
  \end{align}
\end{theorem}
\begin{proof}
  Let $f(s_i) = \text{P}_{S_i}(s_i | r_i = k)$ and $g(s_i) = \text{P}_{S_i}(s_i | r_i = k+1)$ be the splittable functions in (\ref{eq:splittable}).
  \begin{align*}
    & \mathbb{E} \left[ S_i | r_i = k \right] - \mathbb{E} \left[ S_i | r_i = k+1 \right] \\
    &= \int_0^\infty s_i \left[ f(s_i) - g(s_i) \right] ds_i \\
    &= \int_0^v s_i \left[ f(s_i) - g(s_i) \right] ds_i + \int_v^\infty s_i \left[ f(s_i) - g(s_i) \right] ds_i \\
    &\stackrel{\rm (\ref{eq:splittable})}{\geq} \int^v_0 v \left[ f(s_i) - g(s_i) \right] ds_i + \int_v^\infty v \left[ f(s_i) - g(s_i) \right] ds_i \\
    &= v \int_0^\infty \left[ f(s_i) - g(s_i) \right] ds_i = 0.
  \end{align*}
\end{proof}
This indicates that the realization of conditional ranking score $s_i | \{r_i = k\}$ can be changed by the rank of ad
even if $s_i$ is originally an unbiased estimator.
Since the rank $r_i$ depends on the bid prices for all advertisements,
the auction would be non-truthful as mentioned in \S\ref{sec:introduction}.
Lastly, we show that Theorem~\ref{theorem:main_contribution} can be simplified when the ranking scores are mutually independent -- the splittable assumption can be replaced with the independence assumption in the following theorem.
\begin{theorem}
  \label{theorem:main_contribution_independent}
  For any $i \in \{1,\cdots,m-1\}$, $k \in \{1,\cdots,m-1\}$, and non-negative ranking score $s_i \ge 0$,
  if all of the ranking scores $S_i, \forall i \in \{1,\cdots,m\}$ are mutually independent,
  then the following inequality holds
  \begin{align}
    \label{eq:main_contribution}
    \mathbb{E} \left[ S_i | r_i = k \right] \ge \mathbb{E} \left[ S_i | r_i = k+1 \right].
  \end{align}
\end{theorem}
\begin{proof}
  From Bayes' theorem, we have
  \begin{align}
    \label{eq:bayse_transformation}
    & \mathbb{E} \left[ S_i | r_i = k \right] - \mathbb{E} \left[ S_i | r_i = k+1 \right] \nonumber \\
    &= \int_0^\infty s_i \left[ \text{P}_{S_i}(s_i | r_i = k) - \text{P}_{S_i}(s_i | r_i = k + 1) \right] ds_i \nonumber \\
    &= \int_0^\infty s_i \text{P}_{S_i}(s_i) \left[ \frac{\text{P}_{S_i}(r_i=k|s_i)}{\text{P}_{S_i}(r_i=k)} - \frac{\text{P}_{S_i}(r_i=k+1|s_i)}{\text{P}_{S_i}(r_i=k+1)} \right] ds_i \nonumber \\
    &\varpropto \int_0^\infty s_i \text{P}_{S_i}(s_i) \left[ \text{P}_{S_i}(r_i=k|s_i) - \alpha(k) \text{P}_{S_i}(r_i=k+1|s_i) \right] ds_i,
  \end{align}
  where $\alpha(k) = \frac{\text{P}_{S_i}(r_i=k)}{\text{P}_{S_i}(r_i=k+1)} \ge 0$.
  Note that $\alpha(k)$ does not depend on $s_i$
  because $\text{P}_{S_i}(r_i=k)$ and $\text{P}_{S_i}(r_i=k+1)$ are marginalized for all possible $s_i$.
  The basic idea of the proof is to show that (\ref{eq:bayse_transformation}) can be written as
  \begin{align}
    \label{eq:another_transformation}
    (\ref{eq:bayse_transformation})
    = \int_0^\infty s_i \text{P}_{S_i}(s_i) \left[ \phi(s_i) - \psi(s_i) \right] ds_i,
  \end{align}
  where $\phi(s_i)$ and $\psi(s_i)$ are monotonically increasing functions.
  Moreover, these functions have the property that
  \begin{align}
    \label{eq:equal_property_of_functions}
    \int_0^\infty \text{P}_{S_i}(s_i) \left[ \phi(s_i) - \psi(s_i) \right] ds_i = 0.
  \end{align}
  By combining (\ref{eq:equal_property_of_functions}) with the monotonicity of $\phi(s_i)$ and $\psi(s_i)$,
  the equation (\ref{eq:another_transformation}) can be bounded as
  \begin{align}
    \label{eq:lower_bound_bayse_transformation}
    (\ref{eq:bayse_transformation})
    &= \int_0^\infty s_i \text{P}_{S_i}(s_i) \left[ \phi(s_i) - \psi(s_i) \right] ds_i \nonumber \\
    &\ge \int_0^\infty \max(v, 0) \cdot \text{P}_{S_i}(s_i) \left[ \phi(s_i) - \psi(s_i) \right] ds_i = 0,
  \end{align}
  where $v$ is the cross point satisfying $\phi(v) - \psi(v) = 0$.

  Subsequently, we prove (\ref{eq:another_transformation}) and (\ref{eq:equal_property_of_functions}) as follows.
  Let $F_j(s) = \int_{0}^s \text{P}_{S_j}(x) dx = \text{P}(S_j \le s)$ be CDF (cumulative density function) with respect to the non-negative ranking score of the $j$-th ad.
  Because of the assumption that the ranking scores are mutually independent,
  we have the following joint probabilities:
  \begin{align*}
    \text{P}_{S_i}(r_i=1 | s_i)
    &= \prod_{j \in \{1,\cdots,m\} \setminus \{i\}} F_{j}(s_i), \\
    \text{P}_{S_i}(r_i=2 | s_i)
    &= \sum_{\ell \in \{1,\cdots,m\} \setminus \{i\}} [1 - F_{\ell}(s_i)]
       \prod_{j \in \{1,\cdots,m\} \setminus \{i,\ell\}} F_{j}(s_i), \\
  \end{align*}
  where the summation part $\sum_{\ell \in \{1,\cdots,m\} \setminus \{i\}}$ corresponds to the mutual exclusivity that the multiple advertisements can not be assigned to the rank = 1 simultaneously.
  By using this, we define
  \begin{align*}
    \text{P}_{S_i}(r_i=1 | s_i) - \alpha(1) \text{P}_{S_i}(r_i=2 | s_i) = \phi(s_i) - \psi(s_i)
  \end{align*}
  where $\phi(s_i)$ and $\psi(s_i)$ are monotonic functions defined as
  \begin{align*}
    \phi(s_i) &= \!\!\!\!\! \prod_{j \in \{1,\cdots,m\} \setminus \{i\}} \!\!\!\!\! F_{j}(s_i) ~+~
      \alpha(1) \!\!\!\!\!\!\!\!\! \sum_{\ell \in \{1,\cdots,m\} \setminus \{i\}} \!\!\!\!\! F_{\ell}(s_i)
      \!\!\!\!\! \prod_{j \in \{1,\cdots,m\} \setminus \{i,\ell\}} \!\!\!\!\! F_{j}(s_i), \\
    \psi(s_i) &= \alpha(1) \!\!\!\!\!\!\!\!\! \sum_{\ell \in \{1,\cdots,m\} \setminus \{i\}}
       \prod_{j \in \{1,\cdots,m\} \setminus \{i,\ell\}} F_{j}(s_i).
  \end{align*}
  Since $\int_0^\infty \frac{\text{P}_{S_i}(s_i) \text{P}_{S_i}(r_i=k|s_i)}{\text{P}_{S_i}(r_i=k)} ds_i = 1$ for any $k$ in (\ref{eq:bayse_transformation}),
  we have
  \begin{align*}
    & \int_0^\infty \text{P}_{S_i}(s_i) \left[ \text{P}_{S_i}(r_i=1|s_i) - \alpha(1) \text{P}_{S_i}(r_i=2|s_i) \right] ds_i \\
    &=\int_0^\infty \text{P}_{S_i}(s_i) \left[ \phi(s_i) - \psi(s_i) \right] ds_i = 0,
  \end{align*}
  and thus we arrived at (\ref{eq:equal_property_of_functions}).
  Furthermore, for $k=1$, we arrived at (\ref{eq:another_transformation}) as
  \begin{align*}
    (\ref{eq:bayse_transformation})
    =
    \int_0^\infty s_i \text{P}_{S_i}(s_i) \left[ \phi(s_i) - \psi(s_i) \right] ds_i.
  \end{align*}
  We omit the proof for general $k$ but it can be shown by using the same technique.
\end{proof}
Theorem~\ref{theorem:main_contribution_independent} implies that
the functions $f(s_i) = \text{P}_{S_i}(s_i) \phi(s_i)$ and $g(s_i) = \text{P}_{S_i}(s_i) \psi(s_i)$ shown in (\ref{eq:lower_bound_bayse_transformation}) are the splittable functions,
which means that
Theorem~\ref{theorem:main_contribution_independent} is a specialization of Theorem~\ref{theorem:main_contribution}.

%% file: sec4_experiments.tex
\section{Experiments}
\label{sec:experiments}
In this section, we evaluate the effect of selection bias issue on A/B testing for a CTR prediction task.
Our theorem indicates that the effect of the bias can not be canceled in CPC computation
because it depends on the rank of ad.
To show this, we compute the following metrics:
\begin{align}
  \label{eq:relative_calibration}
  C_\text{relative} = \frac{\text{Calibration}_{1}}{\text{Calibration}_\text{rand}},
\end{align}
\begin{align*}
  \text{Calibration}_{1} =
    \frac{
      \sum_{i \in \mathcal{I}_1} \widehat{\text{CTR}}_{i}
    }{
      \sum_{i \in \mathcal{I}_1} \tilde{c}_i
    }, \quad
  \text{Calibration}_\text{rand} =
    \frac{
      \sum_{i \in \mathcal{I}_\text{rand}} \widehat{\text{CTR}}_{i}
    }{
      \sum_{i \in \mathcal{I}_\text{rand}} \tilde{c}_i
    },
\end{align*}
and we explain the definition of $\mathcal{I}_1, \mathcal{I}_\text{rand}$ and $\tilde{c}_i$ step-by-step.
(i) $\mathcal{I}_1$ is the set of indices for displayed advertisements, each rank of which is one.
An example of $\mathcal{I}_1$ is $\mathcal{I}_1 = \{2,2,3\}$
when we have three advertisements: $\text{ad}_1, \text{ad}_2, \text{ad}_3$
where $\text{ad}_2$ is displayed two times and $\text{ad}_3$ is displayed only once.
(ii) $\mathcal{I}_\text{rand}$ is similar to $\mathcal{I}_1$ but the advertisements are randomly selected and displayed.
Since the ranking scores of the advertisements for $\mathcal{I}_\text{rand}$ are relatively small,
we obtain $\mathcal{I}_\text{rand}$ in few cases.
In this paper, we obtain $\mathcal{I}_\text{rand}$ by using $\epsilon$-Greedy algorithm as described later.
(iii) $\tilde{c}_i$ is a click response defined in $\{0, 1\}$. If the $i$-th ad is clicked, then $\tilde{c}_i = 1$, otherwise it is zero.
$\text{C}_{relative}$ is expected to be one if the selection bias does not depend on the rank of ad.
We compute $\text{C}_{relative}$ under the different CTR prediction models as follows.

\subsection{CTR Prediction Models}
We use two different CTR prediction models:
(a) we consider a naive CTR estimator $\widehat{\text{CTR}}_i$ for the $i$-th ad,
which is defined by
\begin{align}
  \label{eq:naive_CTR_estimator}
  \widehat{\text{CTR}}_i(site,pos)
  = \frac{c_i(site,pos)}{n_i(site,pos)},
\end{align}
where $site,pos$ are the site ID of the media depicted in \figurename~\ref{fig:yahoo_news} and the position ID in the site, respectively.
Moreover, $c_i, n_i$ are the number of clicks and impressions for the $i$-th ad, respectively,
which are individually aggregated in each site $\times$ position.
We assume that CTR is in a \emph{steady state} for a few days,
and hence we aggregate $c_i, n_i$ over recent $L$ days;
we experimentally set the window size $L$ as two weeks in our experiment,
and we continually update $c_i, n_i$ until the end of our experiment.
(b) we use the following machine learning based CTR prediction model
\begin{align}
  \label{eq:machine_learning_CTR_predictor}
  \widehat{\text{CTR}}(\bm{x}_\text{user},\bm{x}_{\text{ad}_i}) = \text{P}(click=1|\bm{x}_\text{user},\bm{x}_{\text{ad}_i}; \bm{\Theta}),
\end{align}
where $\bm{\Theta}$ is an optimization parameter,
$\bm{x}_\text{user}$ and $\bm{x}_{\text{ad}_i}$ are the features for a given user and the $i$-th ad.
For concrete example, we can use
\begin{align*}
  \bm{x}_\text{user} &= (\bm{x}_\text{site},\bm{x}_\text{pos},\bm{x}_\text{age},\bm{x}_\text{gender},\bm{x}_\text{access-time}), \\
  \bm{x}_{\text{ad}_i} &= (\widehat{\text{CTR}}_i, c_i, n_i, \bm{x}_{\text{size}_i}, \bm{x}_{\text{style}_i}, \bm{x}_{\text{category}_i}),
\end{align*}
where each vector in $\bm{x}_\text{user}$ is a row vector of categorical features which are represented by one-hot encoding, for instance.
Note that we need not use one-hot encoding if we use a neural network based model whose input layer is a word embedding layer.
To create $\bm{x}_\text{access-time}$, we quantize the user access\footnote{In this paper, the \emph{access} stands for the user request for returning a content of the media.} time to hour and day of the week.
$\widehat{\text{CTR}}_i, c_i, n_i$ are the same as (\ref{eq:naive_CTR_estimator}),
but we use click information before the target label is observed for circumventing the label leakage problem.
The categorical features $\bm{x}_{\text{size}_i}, \bm{x}_{\text{style}_i}, \bm{x}_{\text{category}_i}$
represent schematic informations: size or aspect ratio, style (text, image, video), and category (finance, game, food, and so on).
We use the binary cross entropy loss a.k.a. logarithmic loss to train the optimization parameter $\bm{\Theta}$
which can be updated on a daily basis by using the training data for the latest seven days, for instance.
We omit the details of the features and the machine learning model we used in our experiment,
because it is currently confidential and only for internal use within our company.

The notable difference between the models (a) and (b) is that,
the naive CTR estimator (\ref{eq:naive_CTR_estimator}) is individually computed in each advertisement,
while we train the optimization parameter $\bm{\Theta}$ by using the data for all advertisements;
the latter can be regarded as a multi-task learning.
We checked the logarithmic loss of the two models for prepered test data in advance,
the model (b) was better than (a) about 1.5\% to 2\% (this is empirically not small).
We think that it works well by using a lot of data, although the complexity of model (b) is much greater than (a).

In our experiment, we use $\epsilon$-Greedy algorithm to accumulate the click and impression data.
For each user access,
an advertisement is randomly selected and displayed with a small probability $\epsilon$,
while it is selected by referring ranking scores with the probability $1 - \epsilon$.
The randomly selected advertisements are used in computing $\mathcal{I}_\text{rand}$ as mentioned before.

\subsection{A/B Testing}
To compare the selection bias in each CTR prediction model,
we use an online A/B testing: two independent subsets of users are randomly selected for the comparison, and we call each subset “bucket” with suffix A or B.
We would usually have to divide the advertiser's budgets into small subsets corresponding to each of the buckets, but for this test we do not divide them because the goal of this experiment is to review the selection bias issue in each CTR prediction model -- not to review the behavior of the advertisers.

\subsection{Other Metrics in Online Advertising}
Here, we explain the total utility (TU) of advertisers in our A/B testing,
which can be written as
\begin{align*}
  \text{TU}^{(\text{bucket})} = \sum_{i \in \mathcal{I}_1^{(\text{bucket})}} \tilde{c}_i (v_i - \text{CPC}_i),
\end{align*}
where
$\mathcal{I}_1^{(\text{bucket})}$ is the set of indices for displayed advertisements in the bucket,
$\tilde{c}_i$ is a click response defined in $\{0, 1\}$,
and $v_i$ is the value of click for the advertiser who has the $i$-th ad.
If the auction is truthful, then we can approximately estimate $v_i$ as the original bid price $\text{Bid}_i$
because $\text{Bid}_i$ would be set as the highest price that advertiser is willing to pay for user click.
However, it is not truthful under the selection bias issue, and thus we assume that $\text{Bid}_i \varpropto v_i$.
By using the assumption, we experimentally evaluate the following metrics instead of the total utility:
\begin{align}
  \label{eq:relative_total_value}
  \text{RTV~(relative~total~value)}
  &=
  \frac{
    \sum_{i \in \mathcal{I}_1^{(\text{bucket-B})}} \tilde{c}_i \text{Bid}_i
  }{
    \sum_{i \in \mathcal{I}_1^{(\text{bucket-A})}} \tilde{c}_i \text{Bid}_i
  }, \\
  \label{eq:relative_total_cost}
  \text{RTC~(relative~total~cost)}
  &=
  \frac{
    \sum_{i \in \mathcal{I}_1^{(\text{bucket-B})}} \tilde{c}_i \text{CPC}_i
  }{
    \sum_{i \in \mathcal{I}_1^{(\text{bucket-A})}} \tilde{c}_i \text{CPC}_i
  }.
\end{align}

\subsection{Experimental Results}
The A/B testing was carried out for four weeks, and the data for the last half weeks was used in evaluation
(the first half was excluded).
Table~\ref{tb:experimental_result_calbrations} shows $C_\text{relative}$ defined in (\ref{eq:relative_calibration}),
and we added the Bid-weighted version of them which have $\text{Bid}_i \widehat{\text{CTR}}_{i}$ and $\text{Bid}_i \tilde{c}_i$.
These results imply that the selection bias depends on the rank of ad because $C_\text{relative}$ is greater than one,
moreover, this dependency is relatively small in the machine learning model.
\begin{table}[th]
  \caption{$C_\text{relative}$ in each CTR prediction model.}
  \label{tb:experimental_result_calbrations}
  \centering
  \begin{tabular}{c|c|c|c}
  \hline
  \multicolumn{2}{c|}{naive model in (\ref{eq:naive_CTR_estimator})} & \multicolumn{2}{c}{machine learning model in (\ref{eq:machine_learning_CTR_predictor})} \\ \hline\hline
  non-weighted & Bid-weighted & non-weighted & Bid-weighted \\ \hline
  1.2456 & 1.2972 & 1.1588 & 1.1670 \\ \hline
  \end{tabular}
\end{table}

As for the advertiser utility, we experimentally evaluated RTV (\ref{eq:relative_total_value}) and RTC (\ref{eq:relative_total_cost})
when the bucket-A corresponds to the naive model.
The results are summarized in Table~\ref{tb:experimental_result_RTVC}. RTV was improved by +4\% in the machine learning model, but RTC was also increased.
This means that the machine learning model is robust to the selection bias in CPC computation, and the transaction for online advertising would be more efficient in the machine learning model.
\begin{table}[th]
  \caption{The relative metrics of the total value and cost for advertisers.
  }
  \label{tb:experimental_result_RTVC}
  \centering
  \begin{tabular}{c|c}
  \hline
  RTV (relative total value) & RTC (relative total cost) \\ \hline\hline
  1.0429 & 1.0588 \\ \hline
  \end{tabular}
\end{table}

%% file: sec5_conclusion.tex
\section{Conclusion}
\label{sec:conclusion}
In this paper, we analyzed the relation of selection biases of multiple ordered statistics.
Our analysis shows that the selection bias depends on the rank of advertisement --
this means that the selection bias issue can not be ignored in computing CPC (cost per click) which is calculated by the two ordered statistics.

We reviewed the selection bias issue in online A/B testing for a CTR prediction task.
The naive CTR estimator $\widehat{\text{CTR}}_i = \frac{\text{\#~of~clicks}_i}{\text{\#~of~impressions}_i}$ for the $i$-th ad is sensitive to the selection bias, while a machine learning model using multi-task learning is robust to it.
After the ranking in each auction, we suffer from the selection bias issue even if the estimated or predicted CTRs are originally unbiased,
therefore, it is important that the ordered statistics are robust to it.
In future work, we want to obtain such a robust estimator for each auction, but it would be a challenging task because the number of auctions per second in online advertising is extremely huge.

Since the original bid price is basically optimized by a budget allocation algorithm under the budget constraints,
we want to extend our analysis to the case that the bid price is also a random variable;
however, the selection bias issue would be more critical since the number of random variables is increased.